\documentclass[journal]{IEEEtran}
\usepackage{amsfonts}
\usepackage{amsmath}
\usepackage{amsthm}
\usepackage{amssymb}
\usepackage{graphicx}
\usepackage[T1]{fontenc}
\usepackage{supertabular}
\usepackage{longtable}
\usepackage[usenames,dvipsnames]{color}
\usepackage{bbm}
\usepackage{caption}
\usepackage{fancyhdr}
\usepackage{breqn}
\usepackage{fixltx2e}
\usepackage{capt-of}

\newtheorem{theorem}{Theorem}

\newtheorem{proposition}{Proposition}

\newcommand{\mathsym}[1]{{}}
\newcommand{\unicode}[1]{{}}

\begin{document}

\title{On the Super-Additivity and Estimation Biases of Quantile
Contributions}
\author{ 
\IEEEauthorblockN{Nassim Nicholas Taleb\IEEEauthorrefmark{1},
Raphael Douady\IEEEauthorrefmark{2}} \\
\IEEEauthorblockA{\IEEEauthorrefmark{1}School of Engineering, New York University \\
     \and } 
\IEEEauthorblockA{\IEEEauthorrefmark{2}Riskdata \& C.N.R.S. Paris, Labex ReFi, Centre d'Economie de la Sorbonne
     \and } }
\maketitle

\begin{abstract}
Sample measures of top centile contributions to the total (concentration)
are downward biased, unstable estimators, extremely sensitive to sample size
and concave in accounting for large deviations. It makes them particularly
unfit in domains with power law tails, especially for low values of the
exponent. These estimators can vary over time and increase with the
population size, as shown in this article, thus providing the illusion of
structural changes in concentration. They are also inconsistent under
aggregation and mixing distributions, as the weighted average of
concentration measures for $A$ and $B$ will tend to be lower than that from $%
A\cup B$. In addition, it can be shown that under such fat tails, increases
in the total sum need to be accompanied by increased sample size of the
concentration measurement. We examine the estimation superadditivity and
bias under homogeneous and mixed distributions. 
\end{abstract}

\thanks{Fourth version, Nov 11 2014}
\thispagestyle{fancy} 
\markboth{\textbf{Extreme Risk Initiative ---NYU School of Engineering Working Paper Series}}
\flushbottom

\section{Introduction}

Vilfredo Pareto noticed that 80\% of the land in Italy belonged to 20\% of
the population, and vice-versa, thus both giving birth to the power law
class of distributions and the popular saying 80/20. The self-similarity at
the core of the property of power laws \cite{mandelbrot1960pareto} and \cite%
{mandelbrot1963stable} allows us to recurse and reapply the 80/20 to the
remaining 20\%, and so forth until one obtains the result that the top
percent of the population will own about 53\% of the total wealth.

It looks like such a measure of concentration can be seriously biased,
depending on how it is measured, so it is very likely that the true ratio of
concentration of what Pareto observed, that is, the share of the top
percentile, was closer to 70\%, hence changes year-on-year would drift
higher to converge to such a level from larger sample. In fact, as we will
show in this discussion, for, say wealth, more complete samples resulting
from technological progress, and also larger population and economic growth
will make such a measure converge by increasing over time, for no other
reason than expansion in sample space or aggregate value.

The core of the problem is that, for the class one-tailed fat-tailed random
variables, that is, bounded on the left and unbounded on the right, where
the random variable $X\in \lbrack x_{\min },\infty )$, the in-sample
quantile contribution is a biased estimator of the true value of the actual
quantile contribution.

Let us define the \emph{quantile contribution}%
\begin{equation*}
\kappa _{q}=q\frac{\mathbb{E}[X|X>h(q)]}{\mathbb{E}[X]}
\end{equation*}%
where $h(q)=\inf \{h\in \left[ x_{min},+\infty \right) ,\mathbb{P}(X>h)\leq
q\}$ is the exceedance threshold for the probability $q.$

For a given sample $(X_{k})_{1\leq k\leq n}$, its "natural" estimator $%
\widehat{\kappa }_{q}\equiv \frac{q^{th}\text{percentile}}{total}$, used in
most academic studies, can be expressed, as 
\begin{equation*}
\widehat{\kappa }_{q}\equiv \frac{\sum_{i=1}^{n}\mathbbm{1}_{X_{i}>\hat{h}%
(q)}X_{i}}{\sum_{i=1}^{n}X_{i}}
\end{equation*}%
where $\hat{h}(q)$ is the estimated exceedance threshold for the probability 
$q:$%
\begin{equation*}
\hat{h}(q)=\inf \{h:\frac{1}{n}\sum_{i=1}^{n}\mathbbm{1}_{x>h}\leq q\}
\end{equation*}%
We shall see that the observed variable $\widehat{\kappa }_{q}$ is a
downward biased estimator of the true ratio $\kappa _{q}$, the one that
would hold out of sample, and such bias is in proportion to the fatness of
tails and, for very fat tailed distributions, remains significant, even for
very large samples.

\section{Estimation For Unmixed Pareto-Tailed Distributions}

Let $X$ be a random variable belonging to the class of distributions with a
"power law" right tail, that is: 
\begin{equation}
\mathbb{P}(X>x)\sim L(x)\,x^{-\alpha }  \label{powerlaweq}
\end{equation}%
where $L:\left[ x_{\min },+\infty \right) \rightarrow \left( 0,+\infty
\right) $ is a slowly varying function, defined as $\lim_{x\rightarrow
+\infty }\frac{L(kx)}{L(x)}=1$ for any $k>0$.

There is little difference for small exceedance quantiles (<50\%) between
the various possible distributions such as Student's t, L\'{e}vy $\alpha $%
-stable, Dagum,\cite{dagum1980inequality},\cite{dagum1983income}
Singh-Maddala distribution \cite{singh1978function}, or straight Pareto.

For exponents $1\leq \alpha \leq 2$, as observed in \cite{taleb2014silent},
the law of large numbers operates, though \textit{extremely} slowly. The
problem is acute for $\alpha $ around, but strictly above 1 and severe, as
it diverges, for $\alpha =1$.

\subsection{Bias and Convergence}

\subsubsection{Simple Pareto Distribution}

Let us first consider $\phi _{\alpha }(x)$ the density of a $\alpha $-Pareto
distribution bounded from below by $x_{\min }>0,$ in other words: $\phi
_{\alpha }(x)=\alpha x_{\min }^{\alpha }x^{-\alpha -1}\mathbbm{1}_{x\geq
x_{\min }}$, and $\mathbb{P}(X>x)=\left( \frac{x_{\min }}{x}\right)
{}^{\alpha }$. Under these assumptions, the cutpoint of exceedance is $%
h(q)=x_{\min }\,q^{-1/\alpha }$ and we have:%
\begin{equation}
\kappa _{q}=\frac{\int_{h(q)}^{\infty }x\,\phi (x)dx}{\int_{x_{min}}^{\infty
}x\,\phi (x)dx}=\left( \frac{h(q)}{x_{\min }}\right) {}^{1-\alpha }=q^{\frac{%
\alpha -1}{\alpha }}  \label{qequation}
\end{equation}%
If the distribution of $X$ is $\alpha $-Pareto only beyond a cut-point $x_{%
\text{cut}}$, which we assume to be below $h(q)$, so that we have $\mathbb{P}%
(X>x)=\left( \frac{\lambda }{x}\right) {}^{\alpha }$ for some $\lambda >0$,
then we still have $h(q)=\lambda q^{-1/\alpha }$ and%
\begin{equation*}
\kappa _{q}=\frac{\alpha }{\alpha -1}\frac{\lambda }{\mathbb{E}\left[ X%
\right] }q^{\frac{\alpha -1}{\alpha }}
\end{equation*}%
The estimation of $\kappa _{q}$ hence requires that of the exponent $\alpha $
as well as that of the scaling parameter $\lambda $, or at least its ratio
to the expectation of $X$.

Table \ref{biases} shows the bias of $\widehat{\kappa }_{q}$ as an estimator
of $\kappa _{q}$ in the case of an $\alpha $-Pareto distribution for $\alpha
=1.1$, a value chosen to be compatible with practical economic measures,
such as the wealth distribution in the world or in a particular country,
including developped ones.\footnote{%
This value, which is lower than the estimated exponents one can find in the
literature -- around 2 -- is, following \cite{falk1995testing}, a lower
estimate which cannot be excluded from the observations.} In such a case,
the estimator is extemely sensitive to \emph{"small"} samples, \emph{"small"}
meaning in practice $10^{8}$. We ran up to a trillion simulations across
varieties of sample sizes. While $\kappa _{0.01}\approx 0.657933$, even a
sample size of 100 million remains severely biased as seen in the table.

Naturally the bias is rapidly (and nonlinearly) reduced for $\alpha $
further away from 1, and becomes weak in the neighborhood of 2 for a
constant $\alpha $, though not under a mixture distribution for $\alpha $,
as we shall se later. It is also weaker outside the top 1\% centile, hence
this discussion focuses on the famed "one percent" and on low values of the $%
\alpha $ exponent.

\begin{table}[h]
\caption{Biases of Estimator of $\protect\kappa =0.657933$ From $10^{12}$
Monte Carlo Realizations}
\label{biases}%
\begin{tabular}{l|ccc}
$\widehat{\kappa}(n)$ & Mean & Median & STD \\ 
&  &  & across MC runs \\ \hline
$\widehat{\kappa}(10^3)$ & 0.405235 & 0.367698 & 0.160244 \\ 
$\widehat{\kappa}(10^4)$ & 0.485916 & 0.458449 & 0.117917 \\ 
$\widehat{\kappa}(10^5)$ & 0.539028 & 0.516415 & 0.0931362 \\ 
$\widehat{\kappa}(10^6)$ & 0.581384 & 0.555997 & 0.0853593 \\ 
$\widehat{\kappa}(10^7)$ & 0.591506 & 0.575262 & 0.0601528 \\ 
$\widehat{\kappa}(10^8)$ & 0.606513 & 0.593667 & 0.0461397 \\ 
&  &  & 
\end{tabular}%
\end{table}

In view of these results and of a number of tests we have performed around
them, we can conjecture that the bias $\kappa _{q}-\widehat{\kappa }_{q}(n)$
is "of the order of" $\displaystyle c(\alpha ,q)n^{-b(q)(\alpha -1)}$ where
constants $b(q)$ and $c(\alpha ,q)$ need to be evaluated. Simulations
suggest that $b(q)=1,$ whatever the value of $\alpha $ and $q$, but the
rather slow convergence of the estimator and of its standard deviation to 0
makes precise estimation difficult.

\subsubsection{General Case}

In the general case, let us fix the threshold $h$ and define:%
\begin{equation*}
\kappa _{h}=P(X>h)\frac{\mathbb{E}[X|X>h]}{\mathbb{E}[X]}=\frac{\mathbb{E}[X%
\mathbbm{1}_{X>h}]}{\mathbb{E}[X]}
\end{equation*}%
so that we have $\kappa _{q}=\kappa _{h(q)}.$ We also define the $n$-sample
estimator:%
\begin{equation*}
\widehat{\kappa }_{h}\equiv \frac{\sum_{i=1}^{n}\mathbbm{1}_{X_{i}>h}X_{i}}{%
\sum_{i=1}^{n}X_{i}}
\end{equation*}%
where $X_{i}$ are $n$ independent copies of $X$. The intuition behind the
estimation bias of $\kappa _{q}$ by $\widehat{\kappa }_{q}$ lies in a
difference of concavity of the concentration measure with respect to an
innovation (a new sample value), whether it falls below or above the
threshold. Let $A_{h}(n)=\sum_{i=1}^{n}\mathbbm{1}_{X_{i}>h}X_{i}$ and $%
S(n)=\sum_{i=1}^{n}X_{i},$ so that $\displaystyle\widehat{\kappa }_{h}(n)=%
\frac{A_{h}(n)}{S(n)}$ and assume a frozen threshold $h$. If a new sample
value $X_{n+1}<h$ then the new value is $\displaystyle\widehat{\kappa }%
_{h}(n+1)=\frac{A_{h}(n)}{S(n)+X_{n+1}}.$ The value is convex in $X_{n+1}$
so that uncertainty on $X_{n+1}$ increases its expectation. At variance, if
the new sample value $X_{n+1}>h$, the new value $\widehat{\kappa }%
_{h}(n+1)\approx \frac{A_{h}(n)+X_{n+1}-h}{S(n)+X_{n+1}-h}=1-\frac{%
S(n)-A_{h}(n)}{S(n)+X_{n+1}-h},$ which is now concave in $X_{n+1},$ so that
uncertainty on $X_{n+1}$ reduces its value. The competition between these
two opposite effects is in favor of the latter, because of a higher
concavity with respect to the variable, and also of a higher variability
(whatever its measurement) of the variable conditionally to being above the
threshold than to being below. The fatter the right tail of the
distribution, the stronger the effect. Overall, we find that $\displaystyle%
\mathbb{E}\left[ \widehat{\kappa }_{h}(n)\right] \leq \frac{\mathbb{E}\left[
A_{h}(n)\right] }{\mathbb{E}\left[ S(n)\right] }=\kappa _{h} $ (note that
unfreezing the threshold $\hat{h}(q)$ also tends to reduce the concentration
measure estimate, adding to the effect, when introducing one extra sample
because of a slight increase in the expected value of the estimator $\hat{h}%
(q)$, although this effect is rather negligible). We have in fact the
following:

\begin{proposition}
Let $\boldsymbol{X}=(X)_{i=1}^{n}$ a random sample of size $n>\frac{1}{q}$, $%
Y=X_{n+1}$ an extra single random observation, and define: $\displaystyle%
\widehat{\kappa }_{h}(\boldsymbol{X}\sqcup Y)=\frac{\sum_{i=1}^{n}\mathbbm{1}%
_{X_{i}>h}X_{i}+\mathbbm{1}_{Y>h}Y}{\sum_{i=1}^{n}X_{i}+Y}$. We remark that,
whenever $Y>h$, one has: 
\begin{equation*}
\frac{\partial ^{2}\widehat{\kappa }_{h}(\boldsymbol{X}\sqcup Y)}{\partial
Y^{2}}\leq 0.
\end{equation*}%
This inequality is still valid with $\widehat{\kappa }_{q}$ as the value $%
\hat{h}(q,\boldsymbol{X}\sqcup Y)$ doesn't depend on the particular value of 
$Y>\hat{h}(q,\boldsymbol{X}).$
\end{proposition}

We face a different situation from the common small sample effect resulting
from high impact from the rare observation in the tails that are less likely
to show up in small samples, a bias which goes away by repetition of sample
runs. The concavity of the estimator constitutes a upper bound for the
measurement in finite $n$, clipping large deviations, which leads to
problems of aggregation as we will state below in Theorem 1. 
\begin{figure}[h]
\includegraphics[width=\linewidth]{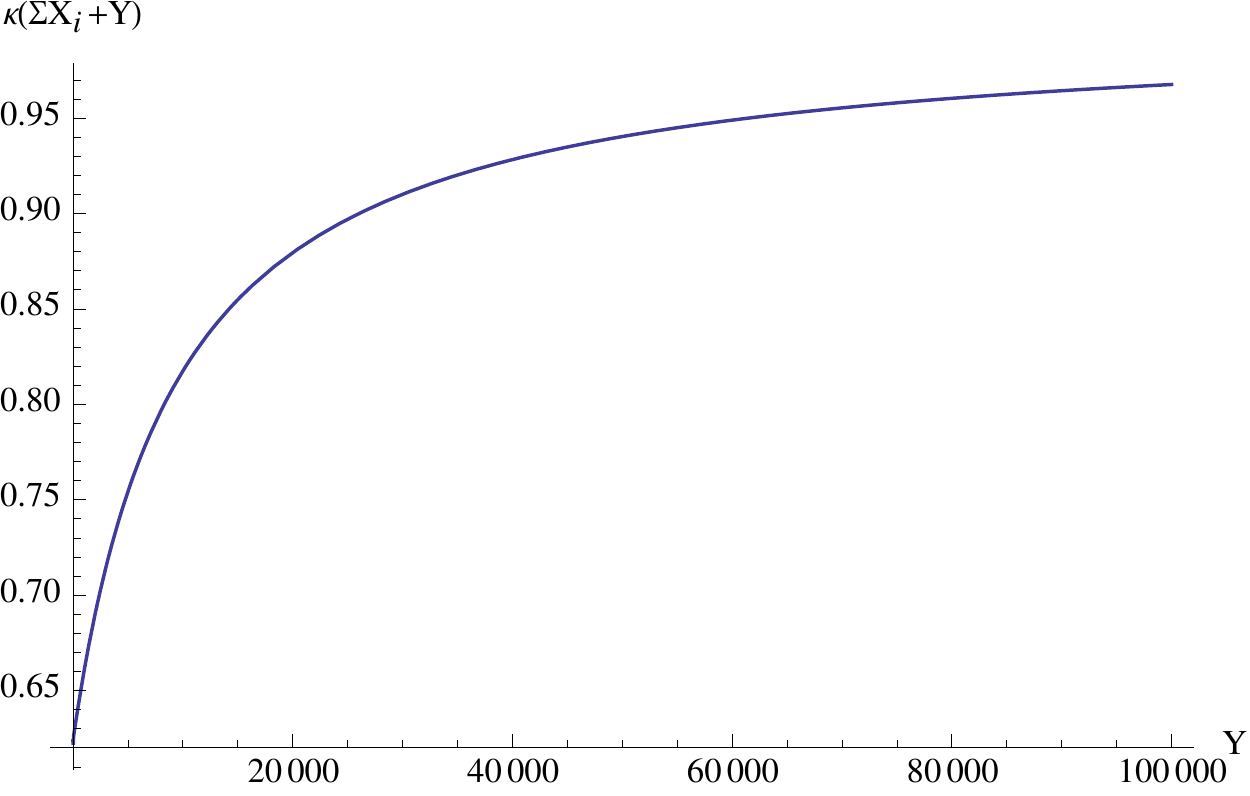}
\caption{Effect of additional observations on $\protect\kappa $}
\end{figure}
\begin{figure}[h]
\includegraphics[width=\linewidth]{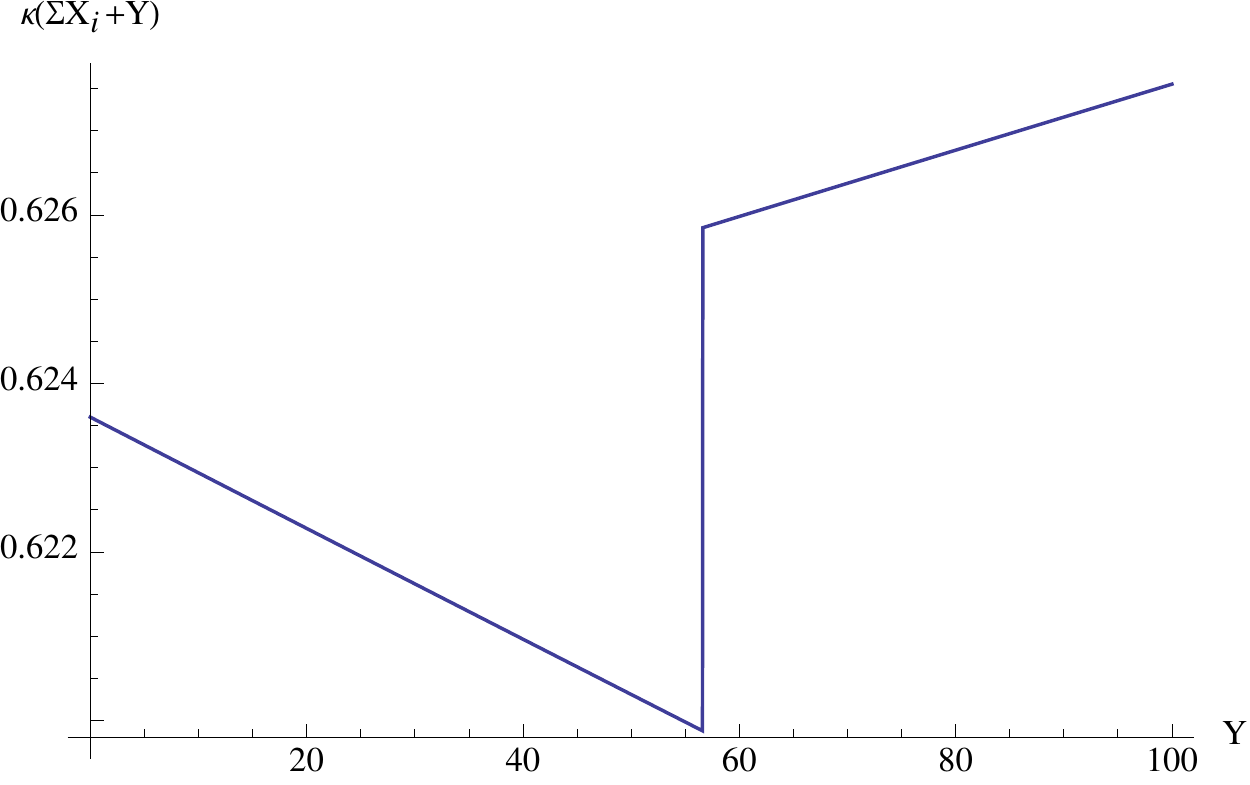}
\caption{Effect of additional observations on $\protect\kappa $, we can see
convexity on both sides of $h$ except for values of no effect to the left of 
$h$, an area of order $1/n$}
\end{figure}
In practice, even in very large sample, the contribution of very large rare
events to $\kappa _{q}$ slows down the convergence of the sample estimator
to the true value. For a better, unbiased estimate, one would need to use a
different path: first estimating the distribution parameters $\left( \hat{%
\alpha},\hat{\lambda}\right) $ and only then, estimating the theoretical
tail contribution $\kappa _{q}(\hat{\alpha},\hat{\lambda})$. Falk \cite%
{falk1995testing} observes that, even with a proper estimator of $\alpha $
and $\lambda $, the convergence is extremely slow, namely of the order of $%
n^{-\delta }/\ln n$, where the exponent $\delta $ depends on $\alpha $ and
on the tolerance of the actual distribution vs. a theoretical Pareto,
measured by the Hellinger distance. In particular, $\delta \rightarrow 0$ as 
$\alpha \rightarrow 1$, making the convergence really slow for low values of 
$\alpha $.

\section{An Inequality About Aggregating Inequality}

For the estimation of the mean of a fat-tailed r.v. $(X)_{i}^{j}$, in $m$
sub-samples of size $n_{i}$ each for a total of $n=\sum_{i=1}^{m}n_{i}$, the
allocation of the total number of observations $n$ between $i$ and $j$ does
not matter so long as the total $n$ is unchanged. Here the allocation of $n$
samples between $m$ sub-samples does matter because of the concavity of $%
\kappa $.\footnote{%
The same concavity -- and general bias -- applies when the distribution is
lognormal, and is exacerbated by high variance.} Next we prove that global
concentration as measured by $\widehat{\kappa }_{q}$ on a broad set of data
will appear higher than local concentration, so aggregating European data,
for instance, would give a $\widehat{\kappa }_{q}$ higher than the average
measure of concentration across countries -- an \emph{"inequality about
inequality"}. In other words, we claim that the estimation bias when using $%
\widehat{\kappa }_{q}(n)$ is even increased when dividing the sample into
sub-samples and taking the weighted average of the measured values $\widehat{%
\kappa }_{q}(n_{i})$.

\begin{theorem}
\label{AggrInt}Partition the $n$ data into $m$ sub-samples $N=N_{1}\cup
\ldots \cup N_{m}$ of respective sizes $n_{1},\ldots ,n_{m}$, with $%
\sum_{i=1}^{m}n_{i}=n$, and let $S_{1},\ldots ,S_{m}$ be the sum of
variables over each sub-sample, and $S=\displaystyle\sum%
\nolimits_{i=1}^{m}S_{i}$ be that over the whole sample. Then we have:%
\begin{equation*}
\mathbb{E}\left[ \widehat{\kappa }_{q}(N)\right] \geq \sum_{i=1}^{m}\mathbb{E%
}\left[ \frac{S_{i}}{S}\right] \mathbb{E}\left[ \widehat{\kappa }_{q}(N_{i})%
\right] 
\end{equation*}%
If we further assume that the distribution of variables $X_{j}$ is the same
in all the sub-samples. Then we have: 
\begin{equation*}
\mathbb{E}\left[ \widehat{\kappa }_{q}(N)\right] \geq \sum_{i=1}^{m}\frac{%
n_{i}}{n}\mathbb{E}\left[ \widehat{\kappa }_{q}(N_{i})\right] 
\end{equation*}
\end{theorem}

In other words, averaging concentration measures of subsamples, weighted by
the total sum of each subsample, produces a downward biased estimate of the
concentration measure of the full sample.

\begin{proof}
An elementary induction reduces the question to the case of two sub-samples.
Let $q\in (0,1)$ and $\left( X_{1},\ldots ,X_{m}\right) $ and $\left(
X_{1}^{\prime },\ldots ,X_{n}^{\prime }\right) $ be two samples of positive
i.i.d. random variables, the $X_{i}$'s having distributions $p(dx)$ and the $%
X_{j}^{\prime }$'s having distribution $p^{\prime }(dx^{\prime }).$ For
simplicity, we assume that both $qm$ and $qn$ are integers. We set $S=%
\displaystyle\sum\limits_{i=1}^{m}X_{i}$ and $S^{\prime }=\displaystyle%
\sum\limits_{i=1}^{n}X_{i}^{\prime }.$ We define $A=\displaystyle%
\sum\limits_{i=1}^{mq}X_{[i]}$ where $X_{[i]}$ is the $i$-th largest value
of $\left( X_{1},\ldots ,X_{m}\right) $, and $A^{\prime }=\displaystyle%
\sum\limits_{i=1}^{mq}X_{[i]}^{\prime }$ where $X_{[i]}^{\prime }$ is the $i$%
-th largest value of $\left( X_{1}^{\prime },\ldots ,X_{n}^{\prime }\right) .
$ We also set $S^{\prime \prime }=S+S^{\prime }$ and $A"=\displaystyle%
\sum\limits_{i=1}^{(m+n)q}X_{[i]}^{\prime \prime }$ where $X_{[i]}^{\prime
\prime }$ is the $i$-th largest value of the joint sample $(X_{1},\ldots
,X_{m},X_{1}^{\prime },\ldots ,X_{n}^{\prime }).$

The $q$-concentration measure for the samples $\boldsymbol{X}%
=(X_{1},...,X_{m}),$ $\boldsymbol{X}^{\prime }=(X_{1}^{\prime
},...,X_{n}^{\prime })$ and $\boldsymbol{X}^{\prime \prime }=(X_{1},\ldots
,X_{m},X_{1}^{\prime },\ldots ,X_{n}^{\prime })$ are:%
\begin{equation*}
\kappa =\frac{A}{S}\qquad \kappa ^{\prime }=\frac{A^{\prime }}{S^{\prime }}%
\qquad \kappa ^{\prime \prime }=\frac{A^{\prime \prime }}{S^{\prime \prime }}
\end{equation*}%
We must prove that he following inequality holds for expected concentration
measures:%
\begin{equation*}
\mathbb{E}\left[ \kappa ^{\prime \prime }\right] \geq \mathbb{E}\left[ \frac{%
S}{S^{\prime \prime }}\right] \mathbb{E}\left[ \kappa \right] +\mathbb{E}%
\left[ \frac{S^{\prime }}{S^{\prime \prime }}\right] \mathbb{E}\left[ \kappa
^{\prime }\right] 
\end{equation*}%
We observe that:%
\begin{equation*}
A=\max_{\substack{ J\subset \left\{ 1,...,m\right\}  \\ \left\vert
J\right\vert =\theta m}}\sum_{i\in J}X_{i}
\end{equation*}%
and, similarly $A^{\prime }=\max_{J^{\prime }\subset \left\{ 1,...,n\right\}
,\left\vert J^{\prime }\right\vert =qn}\sum_{i\in J^{\prime }}X_{i}^{\prime }
$ and $A^{\prime \prime }=\max_{J^{\prime \prime }\subset \left\{
1,...,m+n\right\} ,\left\vert J^{\prime \prime }\right\vert
=q(m+n)}\sum_{i\in J^{\prime \prime }}X_{i},$ where we have denoted $%
X_{m+i}=X_{i}^{\prime }$ for $i=1\ldots n.$ If $J\subset \left\{
1,...,m\right\} ,\left\vert J\right\vert =\theta m$ and $J^{\prime }\subset
\left\{ m+1,...,m+n\right\} ,\left\vert J^{\prime }\right\vert =qn$, then $%
J^{\prime \prime }=J\cup J^{\prime }$ has cardinal $m+n$, hence $A+A^{\prime
}=\sum_{i\in J^{\prime \prime }}X_{i}\leq A^{\prime \prime },$ whatever the
particular sample. Therefore $\kappa ^{\prime \prime }\geq \frac{S}{%
S^{\prime \prime }}\kappa +\frac{S^{\prime }}{S^{\prime \prime }}\kappa
^{\prime }$ and we have: 
\begin{equation*}
\mathbb{E}\left[ \kappa ^{\prime \prime }\right] \geq \mathbb{E}\left[ \frac{%
S}{S^{\prime \prime }}\kappa \right] +\mathbb{E}\left[ \frac{S^{\prime }}{%
S^{\prime \prime }}\kappa ^{\prime }\right] 
\end{equation*}

Let us now show that:\quad 
\begin{equation*}
\mathbb{E}\left[ \frac{S}{S^{\prime \prime }}\kappa \right] =\mathbb{E}\left[
\frac{A}{S^{\prime \prime }}\right] \geq \mathbb{E}\left[ \frac{S}{S^{\prime
\prime }}\right] \mathbb{E}\left[ \frac{A}{S}\right] 
\end{equation*}%
If this is the case, then we identically get for $\kappa ^{\prime }:$%
\begin{equation*}
\mathbb{E}\left[ \frac{S^{\prime }}{S^{\prime \prime }}\kappa ^{\prime }%
\right] =\mathbb{E}\left[ \frac{A^{\prime }}{S^{\prime \prime }}\right] \geq 
\mathbb{E}\left[ \frac{S^{\prime }}{S^{\prime \prime }}\right] \mathbb{E}%
\left[ \frac{A^{\prime }}{S^{\prime }}\right] 
\end{equation*}%
hence we will have:%
\begin{equation*}
\mathbb{E}\left[ \kappa ^{\prime \prime }\right] \geq \mathbb{E}\left[ \frac{%
S}{S^{\prime \prime }}\right] \mathbb{E}\left[ \kappa \right] +\mathbb{E}%
\left[ \frac{S^{\prime }}{S^{\prime \prime }}\right] \mathbb{E}\left[ \kappa
^{\prime }\right] 
\end{equation*}

Let $T=X_{[mq]}$ be the cut-off point (where $\left[ mq\right] $ is the
integer part of $mq$), so that $A=$ $\displaystyle\sum\limits_{i=1}^{m}X_{i}%
\mathbbm{1}_{X_{i}\geq T}$ and let $B=S-A=\displaystyle\sum%
\limits_{i=1}^{m}X_{i}\mathbbm{1}_{X_{i}<T}.$ Conditionally to $T$, $A$ and $%
B$ are independent: $A$ is a sum if $m\theta $ samples constarined to being
above $T$, while $B$ is the sum of $m(1-\theta )$ independent samples
constrained to being below $T$. They are also independent of $S^{\prime }.$
Let $p_{A}(t,da)$ and $p_{B}(t,db)$ be the distribution of $A$ and $B$
respectively, given $T=t$. We recall that $p^{\prime }(ds^{\prime })$ is the
distribution of $S^{\prime }$ and denote $q(dt)$ that of $T$. We have:%
\begin{multline*}
\mathbb{E}\left[ \frac{S}{S^{\prime \prime }}\kappa \right] = \\
\iint \frac{a+b}{a+b+s^{\prime }}\frac{a}{a+b}p_{A}(t,da)\,p_{B}(t,db)%
\,q(dt)\,p^{\prime }(ds^{\prime })
\end{multline*}

For given $b$, $t$ and $s^{\prime }$, $a\rightarrow \frac{a+b}{a+b+s^{\prime
}}$ and $a\rightarrow \frac{a}{a+b}$ are two increasing functions of the
same variable $a$, hence conditionally to $T$, $B$ and $S^{\prime }$, we
have:%
\begin{multline*}
\mathbb{E}\left[ \left. \frac{S}{S^{\prime \prime }}\kappa \right\vert
T,B,S^{\prime }\right] =\mathbb{E}\left[ \left. \frac{A}{A+B+S^{\prime }}%
\right\vert T,B,S^{\prime }\right]  \\
\geq \mathbb{E}\left[ \left. \frac{A+B}{A+B+S^{\prime }}\right\vert
T,B,S^{\prime }\right] \mathbb{E}\left[ \left. \frac{A}{A+B}\right\vert
T,B,S^{\prime }\right] 
\end{multline*}%
This inequality being valid for any values of $T$, $B$ and $S^{\prime }$, it
is valid for the unconditional expectation, and we have:%
\begin{equation*}
\mathbb{E}\left[ \frac{S}{S^{\prime \prime }}\kappa \right] \geq \mathbb{E}%
\left[ \frac{S}{S^{\prime \prime }}\right] \mathbb{E}\left[ \frac{A}{S}%
\right] 
\end{equation*}

If the two samples have the same distribution, then we have:%
\begin{equation*}
\mathbb{E}\left[ \kappa ^{\prime \prime }\right] \geq \frac{m}{m+n}\mathbb{E}%
\left[ \kappa \right] +\frac{n}{m+n}\mathbb{E}\left[ \kappa ^{\prime }\right]
\end{equation*}%
Indeed, in this case, we observe that $\mathbb{E}\left[ \frac{S}{S^{\prime
\prime }}\right] =\frac{m}{m+n}.$ Indeed $S=\sum_{i=1}^{m}X_{i}$ and the $%
X_{i}$ are identically distributed, hence $\mathbb{E}\left[ \frac{S}{%
S^{\prime \prime }}\right] =m\mathbb{E}\left[ \frac{X}{S^{\prime \prime }}%
\right] .$ But we also have $\mathbb{E}\left[ \frac{S^{\prime \prime }}{%
S^{\prime \prime }}\right] =1=(m+n)\mathbb{E}\left[ \frac{X}{S^{\prime
\prime }}\right] $ therefore $\mathbb{E}\left[ \frac{X}{S^{\prime \prime }}%
\right] =\frac{1}{m+n}$. Similarly, $\mathbb{E}\left[ \frac{S^{\prime }}{%
S^{\prime \prime }}\right] =\frac{n}{m+n},$ yielding the result.

This ends the proof of the theorem.
\end{proof}

Let $X$ be a positive random variable and $h\in (0,1).$ We remind the
theoretical $h$-concentration measure, defined as:%
\begin{equation*}
\kappa _{h}=\frac{P(X>h)\mathbb{E}\left[ X\left\vert X>h\right. \right] }{%
\mathbb{E}\left[ X\right] }
\end{equation*}%
whereas the $n$-sample $\theta $-concentration measure is $\widehat{\kappa }%
_{h}(n)=\frac{A(n)}{S(n)},$ where $A(n)$ and $S(n)$ are defined as above for
an $n$-sample $\boldsymbol{X}=\left( X_{1},\ldots ,X_{n}\right) $ of i.i.d.
variables with the same distribution as $X$.

\begin{theorem}
For any $n\in \mathbb{N},$ we have:%
\begin{equation*}
\mathbb{E}\left[ \widehat{\kappa }_{h}(n)\right] <\kappa _{h}
\end{equation*}%
and%
\begin{equation*}
\lim_{n\rightarrow +\infty }\widehat{\kappa }_{h}(n)=\kappa _{h}\quad \text{%
a.s. and in probability}
\end{equation*}
\end{theorem}

\begin{proof}
The above corrolary shows that the sequence $n\mathbb{E}\left[ \widehat{%
\kappa }_{h}(n)\right] $ is super-additive, hence $\mathbb{E}\left[ \widehat{%
\kappa }_{h}(n)\right] $ is an increasing sequence. Moreover, thanks to the
law of large numbers, $\frac{1}{n}S(n)$ converges almost surely and in
probability to $\mathbb{E}\left[ X\right] $ and $\frac{1}{n}A(n)$ converges
almost surely and in probability to $\mathbb{E}\left[ X\mathbbm{1}_{X>h}%
\right] =P(X>h)\mathbb{E}\left[ X\left\vert X>h\right. \right] $, hence
their ratio also converges almost surely to $\kappa _{h}$. On the other
hand, this ratio is bounded by 1. Lebesgue dominated convergence theorem
concludes the argument about the convergence in probability.
\end{proof}

\section{Mixed Distributions For The Tail Exponent}

Consider now a random variable $X$, the distribution of which $p(dx)$ is a
mixture of parametric distributions with different values of the parameter: $%
p(dx)=\sum_{i=1}^{m}\omega _{i}p_{\alpha _{i}}(dx).$ A typical $n$-sample of 
$X$ can be made of $n_{i}=\omega _{i}n$ samples of $X_{\alpha _{i}}$ with
distribution $p_{\alpha _{i}}.$ The above theorem shows that, in this case,
we have:%
\begin{equation*}
\mathbb{E}\left[ \widehat{\kappa }_{q}(n,X)\right] \geq \displaystyle%
\sum\limits_{i=1}^{m}\mathbb{E}\left[ \frac{S(\omega _{i}n,X_{\alpha _{i}})}{%
S(n,X)}\right] \mathbb{E}\left[ \widehat{\kappa }_{q}(\omega _{i}n,X_{\alpha
_{i}})\right]
\end{equation*}%
When $n\rightarrow +\infty ,$ each ratio $\displaystyle\frac{S(\omega
_{i}n,X_{\alpha _{i}})}{S(n,X)}$ converges almost surely to $\omega _{i}$
respectively, therefore we have the following convexity inequality: 
\begin{equation*}
\kappa _{q}(X)\geq \displaystyle\sum\limits_{i=1}^{m}\omega _{i}\kappa
_{q}(X_{\alpha _{i}})
\end{equation*}

The case of Pareto distribution is particularly interesting. Here, the
parameter $\alpha $ represents the tail exponent of the distribution. If we
normalize expectations to $1$, the cdf of $X_{\alpha }$ is $F_{\alpha
}(x)=1-\left( \frac{x}{x_{\min }}\right) ^{-\alpha }$ and we have:%
\begin{equation*}
\kappa _{q}(X_{\alpha })=q^{\frac{\alpha -1}{\alpha }}
\end{equation*}%
and%
\begin{equation*}
\frac{d^{2}}{d\alpha ^{2}}\kappa _{q}(X_{\alpha })=q^{\frac{\alpha -1}{%
\alpha }}\frac{(\log q)^{2}}{\alpha ^{3}}>0
\end{equation*}%
Hence $\kappa _{q}(X_{\alpha })$ is a convex function of $\alpha $ and we
can write:%
\begin{equation*}
\kappa _{q}(X)\geq \displaystyle\sum\limits_{i=1}^{m}\omega _{i}\kappa
_{q}(X_{\alpha _{i}})\geq \kappa _{q}(X_{\bar{\alpha}})
\end{equation*}%
where $\bar{\alpha}=\sum_{i=1}^{m}\omega _{i}\alpha $.

Suppose now that $X$ is a positive random variable with unknown
distribution, except that its tail decays like a power low with unknown
exponent. An unbiased estimation of the exponent, with necessarily some
amount of uncertainty (i.e., a distribution of possible true values around
some average), would lead to a downward biased estimate of $\kappa _{q}.$

\bigskip Because the concentration measure only depends on the tail of the
distribution, this inequality also applies in the case of a mixture of
distributions with a power decay, as in Equation \ref{powerlaweq}:

\begin{equation}
\mathbb{P}(X>x)\sim \sum_{j=1}^{N}\omega _{i}L_{i}(x)x^{-\alpha {_{j}}}
\label{mixedd}
\end{equation}

The slightest uncertainty about the exponent increases the concentration
index. One can get an actual estimate of this bias by considering an average 
$\bar{\alpha}>1$ and two surrounding values $\alpha ^{+}=\alpha +\delta $
and $\alpha ^{-}=\alpha -\delta .$ The convexity inequaly writes as follows:%
\begin{equation*}
\kappa _{q}(\bar{\alpha})=q^{1-\frac{1}{\bar{\alpha}}}<\frac{1}{2}\left(
q^{1-\frac{1}{\alpha +\delta }}+q^{1-\frac{1}{\alpha -\delta }}\right) 
\end{equation*}

So in practice, an estimated $\bar{\alpha}$ of around $3/2$, sometimes
called the "half-cubic" exponent, would produce similar results as value of $%
\alpha $ much closer ro 1, as we used in the previous section. Simply $%
\kappa _{q}(\alpha )$ is convex, and dominated by the second order effect $%
\frac{\ln (q)q^{1-\frac{1}{\alpha +\delta }}(\ln (q)-2(\alpha +\delta ))}{%
(\alpha +\delta )^{4}}$, an effect that is exacerbated at lower values of $%
\alpha $.

To show how unreliable the measures of inequality concentration from
quantiles, consider that a standard error of 0.3 in the measurement of $%
\alpha $ causes $\kappa _{q}(\alpha )$ to rise by 0.25.

\section{A Larger Total Sum is Accompanied by Increases in $\widehat{\protect%
\kappa }_{q}$}

There is a large dependence between the estimator $\widehat{\kappa }_{q}$
and the sum $S=\displaystyle\sum\limits_{j=1}^{n}X_{j}:$ conditional on an
increase in $\widehat{\kappa }_{q}$ the expected sum is larger. Indeed, as
shown in theorem \ref{AggrInt}, $\widehat{\kappa }_{q}$ and $S$ are
positively correlated.

For the case in which the random variables under concern are wealth, we
observe as in Figure \ref{condineq} such conditional increase; in other
words, since the distribution is of the class of fat tails under
consideration, the maximum is of the same order as the sum, additional
wealth means more measured inequality. Under such dynamics, is quite absurd
to assume that additional wealth will arise from the bottom or even the
middle. (The same argument can be applied to wars, epidemics, size or companies, etc.)

\begin{figure}[h]
\includegraphics[width=\linewidth]{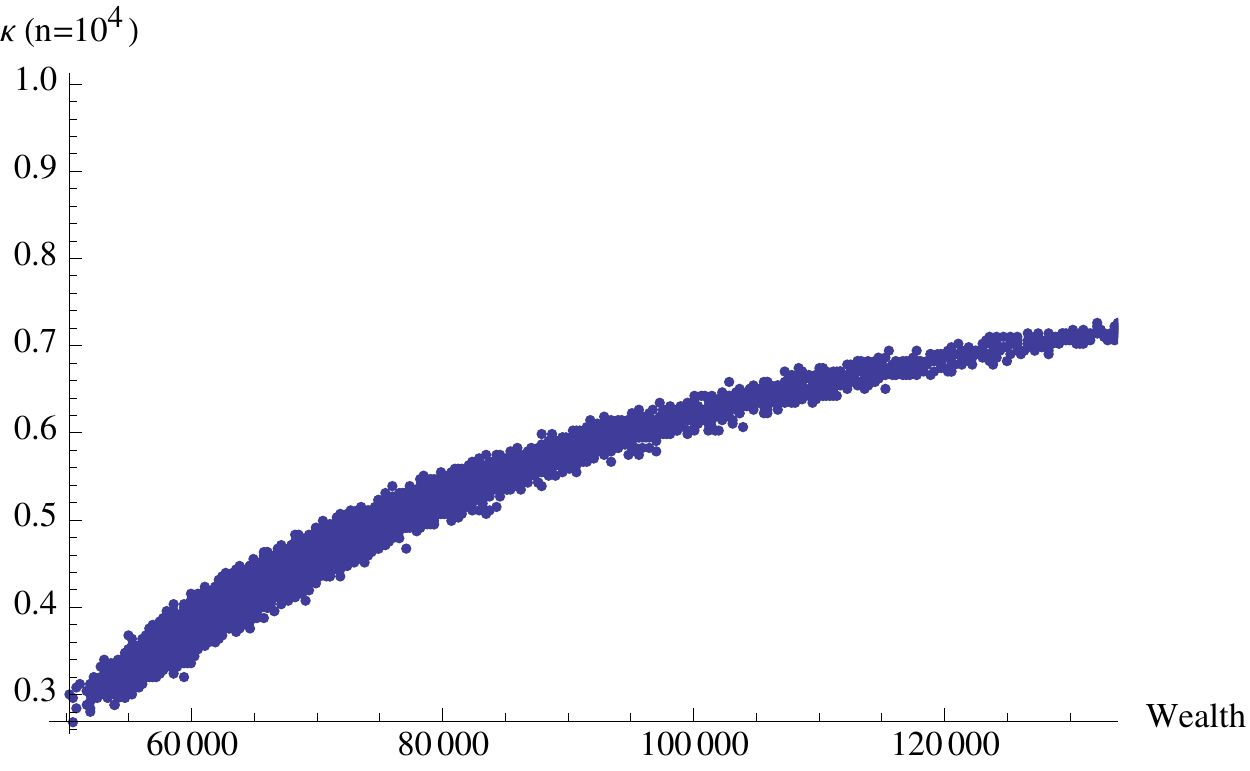}
\caption{Effect of additional wealth on $\hat{\protect\kappa}$}
\label{condineq}
\end{figure}

\section{Conclusion and Proper Estimation of Concentration}

Concentration can be high at the level of the generator, but in small units
or subsections we will observe a lower $\kappa _{q}$. So examining times
series, we can easily get a historical illusion of rise in, say, wealth
concentration when it has been there all along at the level of the process;
and an expansion in the size of the unit measured can be part of the
explanation.\footnote{%
Accumulated wealth is typically thicker tailed than income, see \cite%
{gabaix2008power}.}

Even the estimation of $\alpha $ can be biased in some domains where one
does not see the entire picture: in the presence of uncertainty about the
"true" $\alpha $, it can be shown that, unlike other parameters, the one to
use is not the probability-weighted exponents (the standard average) but
rather the minimum across a section of exponents \cite{taleb2014silent}.

One must not perform analyses of year-on-year changes in $\widehat{\kappa }%
_{q}$ without adjustment. It did not escape our attention that some theories
are built based on claims of such "increase" in inequality, as in \cite%
{piketty2014capital}, without taking into account the true nature of $\kappa
_{q}$, and promulgating theories about the "variation" of inequality without
reference to the stochasticity of the estimation $-$ and the lack of
consistency of $\kappa _{q}$ across time and sub-units. What is worse,
rejection of such theories also ignored the size effect, by countering with
data of a different sample size, effectively making the dialogue on
inequality uninformational statistically.\footnote{\textit{Financial Times},
May 23, 2014 "Piketty findings undercut by errors" by Chris Giles.}

The mistake appears to be commonly made in common inference about fat-tailed
data in the literature. The very methodology of using concentration and
changes in concentration is highly questionable. For instance, in the thesis
by Steven Pinker \cite{pinker2011better} that the world is becoming less
violent, we note a fallacious inference about the concentration of damage
from wars from a $\widehat{\kappa }_{q}$ with minutely small population in
relation to the fat-tailedness.\footnote{%
Using Richardson's data, \cite{pinker2011better}: "(Wars) followed an 80:2
rule: almost eighty percent of the deaths were caused by \textit{two percent}
(his emph.) of the wars". So it appears that both Pinker and the literature
cited for the quantitative properties of violent conflicts are using a
flawed methodology, one that produces a severe bias, as the centile
estimation has extremely large biases with fat-tailed wars. Furthermore claims about the mean become spurious at low exponents.} Owing to the
fat-tailedness of war casualties and consequences of violent conflicts, an
adjustment would rapidly invalidate such claims that violence from war has
statistically experienced a decline.

\subsection{Robust methods and use of exhaustive data}

We often face argument of the type "the method of measuring concentration
from quantile contributions $\hat{\kappa}$ is robust and based on a complete
set of data". Robust methods, alas, tend to fail with fat-tailed data, see 
\cite{taleb2014silent}. But, in addition, the problem here is worse: even if
such "robust" methods were deemed unbiased, a method of direct centile
estimation is still linked to a static and specific population and does not
aggregage. Accordingly, such techniques do not allow us to make statistical
claims or scientific statements about the true properties which should
necessarily carry out of sample.

Take an insurance (or, better, reinsurance) company. The "accounting"
profits in a year in which there were few claims do not reflect on the
"economic" status of the company and it is futile to make statements on the
concentration of losses per insured event based on a single year sample. The
"accounting" profits are not used to predict variations year-on-year, rather
the exposure to tail (and other) events, analyses that take into account the
stochastic nature of the performance. This difference between "accounting"
(deterministic) and "economic" (stochastic) values matters for policy
making, particularly under fat tails. The same with wars: we do not estimate
the severity of a (future) risk based on past in-sample historical data.

\subsection{How Should We Measure Concentration?}

Practitioners of risk managers now tend to compute CVaR and other metrics, methods that are
extrapolative and nonconcave, such as the information from the $\alpha $
exponent, taking the one closer to the lower bound of the range of
exponents, as we saw in our extension to Theorem 2 and rederiving the corresponding $\kappa$%
, or, more rigorously, integrating the functions of $\alpha$ across the various possible
states. Such methods of adjustment are less biased and do not get mixed up
with problems of aggregation --they are similar to the "stochastic
volatility" methods in mathematical finance that consist in adjustments to
option prices by adding a "smile" to the standard deviation, in proportion to the variability of the
parameter representing volatility and the errors in its measurement. Here it
would be "stochastic alpha" or "stochastic tail exponent"\footnote{%
Also note that, in addition to the centile estimation problem, some authors
such as \cite{piketty2006evolution} when dealing with censored data, use
Pareto interpolation for unsufficient information about the tails (based on
tail parameter), filling-in the bracket with conditional average bracket
contribution, which is not the same thing as using full power-law extension;
such a method retains a significant bias.} By extrapolative, we mean the
built-in extension of the tail in the measurement by taking into account
realizations outside the sample path that are in excess of the extrema
observed.\footnote{%
Even using a lognormal distribution, by fitting the scale parameter, works
to some extent as a rise of the standard deviation extrapolates probability
mass into the right tail.} \footnote{We also note that the theorems would also apply to Poisson jumps, but we focus on the powerlaw case in the application, as the methods for fitting Poisson jumps are interpolative and have proved to be easier to fit in-sample than out of sample, see \cite{taleb2014silent}.}

\section*{Acknowledgment}

The late Beno\^{\i}t Mandelbrot, Branko Milanovic, Dominique Gu\'{e}guan,
Felix Salmon, Bruno Dupire, the late Marc Yor, Albert Shiryaev, an anonymous
referee, the staff at Luciano Restaurant in Brooklyn and Naya in Manhattan.

\bibliographystyle{IEEEtran}
\bibliography{/Users/nntaleb/Dropbox/Central-bibliography}

\begin{thebibliography}{10}
\providecommand{\url}[1]{#1}
\csname url@samestyle\endcsname
\providecommand{\newblock}{\relax}
\providecommand{\bibinfo}[2]{#2}
\providecommand{\BIBentrySTDinterwordspacing}{\spaceskip=0pt\relax}
\providecommand{\BIBentryALTinterwordstretchfactor}{4}
\providecommand{\BIBentryALTinterwordspacing}{\spaceskip=\fontdimen2\font plus
\BIBentryALTinterwordstretchfactor\fontdimen3\font minus
  \fontdimen4\font\relax}
\providecommand{\BIBforeignlanguage}[2]{{%
\expandafter\ifx\csname l@#1\endcsname\relax
\typeout{** WARNING: IEEEtran.bst: No hyphenation pattern has been}%
\typeout{** loaded for the language `#1'. Using the pattern for}%
\typeout{** the default language instead.}%
\else
\language=\csname l@#1\endcsname
\fi
#2}}
\providecommand{\BIBdecl}{\relax}
\BIBdecl

\bibitem{mandelbrot1960pareto}
B.~Mandelbrot, ``The pareto-levy law and the distribution of income,''
  \emph{International Economic Review}, vol.~1, no.~2, pp. 79--106, 1960.

\bibitem{mandelbrot1963stable}
------, ``The stable paretian income distribution when the apparent exponent is
  near two,'' \emph{International Economic Review}, vol.~4, no.~1, pp.
  111--115, 1963.

\bibitem{dagum1980inequality}
C.~Dagum, ``Inequality measures between income distributions with
  applications,'' \emph{Econometrica}, vol.~48, no.~7, pp. 1791--1803, 1980.

\bibitem{dagum1983income}
------, \emph{Income distribution models}.\hskip 1em plus 0.5em minus
  0.4em\relax Wiley Online Library, 1983.

\bibitem{singh1978function}
S.~Singh and G.~Maddala, ``A function for size distribution of incomes:
  reply,'' \emph{Econometrica}, vol.~46, no.~2, 1978.

\bibitem{taleb2014silent}
N.~N. Taleb, ``Silent risk: Lectures on fat tails,(anti) fragility, and
  asymmetric exposures,'' \emph{Available at SSRN 2392310}, 2014.

\bibitem{falk1995testing}
M.~Falk \emph{et~al.}, ``On testing the extreme value index via the
  pot-method,'' \emph{The Annals of Statistics}, vol.~23, no.~6, pp.
  2013--2035, 1995.

\bibitem{gabaix2008power}
X.~Gabaix, ``Power laws in economics and finance,'' National Bureau of Economic
  Research, Tech. Rep., 2008.

\bibitem{piketty2014capital}
T.~Piketty, ``Capital in the 21st century,'' 2014.

\bibitem{pinker2011better}
S.~Pinker, \emph{The better angels of our nature: Why violence has
  declined}.\hskip 1em plus 0.5em minus 0.4em\relax Penguin, 2011.

\bibitem{piketty2006evolution}
T.~Piketty and E.~Saez, ``The evolution of top incomes: a historical and
  international perspective,'' National Bureau of Economic Research, Tech.
  Rep., 2006.

\end{thebibliography}

\end{document}